\documentclass[12pt]{article}

\usepackage{amsfonts,amscd,amssymb}
\usepackage{amsthm,amsmath,natbib}
\usepackage{bbm} 
\usepackage{times,graphicx}
\usepackage{mathrsfs}
\usepackage{enumerate}
\usepackage{subfigure}
\usepackage{appendix}
\usepackage{pdfpages}
\usepackage[margin=1in]{geometry}
\usepackage{setspace}
\usepackage{tikz}
\usepackage{todonotes}

\RequirePackage[colorlinks,citecolor=blue,urlcolor=blue]{hyperref}

\usepackage{natbib}

\newtheorem{thm}{Theorem}

\newtheorem{cor}{Corollary}

\newtheorem{Lemma}{Lemma}[section]
\newtheorem{Remark}{Remark}[section]
\newtheorem{Assumption}{Assumption}[section]

\newcommand{\mc}{\mathcal}
\DeclareMathOperator*{\argmax}{argmax}

\newcommand{\tree}{\mathbb{T}}
\newcommand{\leaves}{\mathscr{L}}

\newcommand{\h}{\hspace*{.24in}}

\title{Convergence of maximum likelihood supertree reconstruction}
\date{}
\author{
Lam Si Tung Ho\thanks{These authors contributed equally to this work.} \\
Department of Mathematics and Statistics \\
Dalhousie University, Halifax, Nova Scotia, Canada
\and
Vu Dinh$^{*}$ \\
Department of Mathematical Sciences \\
University of Delaware
}

\begin{document}
\maketitle

\clearpage

\begin{abstract}
Supertree methods are tree reconstruction techniques that combine several smaller gene trees (possibly on different sets of species) to build a larger species tree.
The question of interest is whether the reconstructed supertree converges to the true species tree as the number of gene trees increases (that is, the consistency of supertree methods).
In this paper, we are particularly interested in the convergence rate of the maximum likelihood supertree.
Previous studies on the maximum likelihood supertree approach often formulate the question of interest as a discrete problem and focus on reconstructing the correct topology of the species tree. 
Aiming to reconstruct both the topology and the branch lengths of the species tree, we propose an analytic approach for analyzing the convergence of the maximum likelihood supertree method.
Specifically, we consider each tree as one point of a metric space and prove that the distance between the maximum likelihood supertree and the species tree converges to zero at a polynomial rate under some mild conditions.
We further verify these conditions for the popular exponential error model of gene trees.
\end{abstract}

\section{Introduction}

High-throughput sequencing is making large collections of sequences available to researchers at a low cost.
These genomic data represent a broad spectrum of life and motivates studies of the problem of reconstructing large phylogenetic trees using statistical methods.
Those data, however, also come from different sources, cover different genomic regions on which the evolutionary processes happen very differently, and may not be collected on the same set of species. 
Thus, combining trees on different, overlapping sets of species into a ``supertree" has become a popular approach for reconstructing large species trees.
Over the years, several methods of supertree reconstruction have been developed \citep{cotton2007majority}, and analyses of supertrees continue to play more and more important roles in the search for answers of many fundamental evolutionary questions.

While supertree methods have been of great interest in phylogenetics, little is known about their theoretical properties, especially in non-asymptotic settings. 
In cases when the individual trees are gene trees and the set of taxa are the same across the input trees, several statistical consistent methods have been derived \citep{ mossel2008incomplete, heled2009bayesian, kubatko2009stem, larget2010bucky, liu2010maximum, liu2011estimating, bryant2012inferring, mirarab2014astral, chifman2014quartet, vachaspati2015astrid}. 
However, most proofs of statistical consistency have been analyzed under the condition that each of the gene trees can be estimated accurately, and do not provide a guarantee of robustness to gene tree estimation errors  \citep{roch2015robustness}.
This is a critical concern because when the species tree is constructed using sequences taken from large genomic regions, those regions have a high chance of involving some recombination, which violates one of the main assumptions of the multi-species coalescent model.
On the other hand, limiting analyses to short regions increase gene tree estimation error (GTEE) and various summary methods had impaired accuracy when the error was high \citep{gatesy2014phylogenetic, mirarab2014statistical}.
It is later proved in \cite{roch2019long} that when the sequence length of each locus is bounded and the gene tree cannot be estimated reliably, most summary methods that estimate the species tree by combining gene trees are not statistically consistent.

Because some coalescent-based summary methods sometimes produce less accurate estimates than concatenation \citep{bayzid2013naive, patel2013error, mirarab2014statistical}, seemingly as a result of GTEE, the question of whether provable guarantees can be established in the presence of GTEE (for both species trees and supertree) naturally arises. 
In the context of species tree estimation, \cite{roch2015robustness} establish statistical consistency of the Rooted Triplet Consensus method and the Maximum Pseudo-likelihood for Estimating Species Trees method \citep{liu2010maximum} under GTEE and provide bounds on the sampling complexity for these methods to construct the correct species tree with high probability. 

The maximum likelihood (ML) supertree method is proposed by \cite{steel2008maximum} based on a probability model that permits ``errors" in gene tree topologies and allows the species tree to be estimated even if there is topological conflict amongst gene trees.
\cite{steel2008maximum} shows that ML estimate of the species tree is topologically consistent under fairly general conditions and also shows that the method of Matrix Representation with Parsimony \citep{baum1992combining} may be inconsistent under these same conditions.
However, while the ML estimate is topologically consistent, no results on the convergence rate of the estimator are obtained. 
In this paper, we propose a new analytic approach to study the convergence of the ML supertree method. 
Based on embedding trees into a metric space, we establish the conditions for which the convergence rate of the ML supertree can be obtained. 
We verify these conditions for the popular exponential error model of gene trees, thus obtain a polynomial convergence rate for the ML supertree to the true species tree under this model.

\section{Mathematical framework}

In this paper, the term \emph{phylogenetic tree} refers to a tree $\tree$ with leaves labeled by a set of species.
Each branch of $\tree$ is associated with a non-negative branch length.
A tree is said to be resolved if it is bifurcating and all branch lengths are positive.
Given an unrooted phylogenetic tree $\tree$ on a finite set $\leaves$ of species, any subset $\leaves'$ of $\leaves$ induces a phylogenetic tree on $\leaves'$, denoted $\tree |_{\leaves'}$, which is the subtree of $\tree$ that connects the species in $\leaves'$ only.

We regard the tree space as a metric space $(\mathcal{T}, d)$ where $\mathcal{T}$ is the set of all phylogenetic trees with branch lengths bounded from above by a positive constant $g_0$ and $d$ is a continuous metric. 
For simplicity of presentation, we will assume that $d$ is the BHV distance on the set of trees with $n$ species \citep{billera2001geometry}, but note that the analysis of the paper can be extended to any continuous and locally-Euclidean distance, including the branch-score distance \citep{kuhner1994simulation} and the AGPS distance \citep{amenta2007approximating}.

To describe trees that are ``near" to each other, the BHV distance use the class of nearest neighbor interchange (NNI) moves \citep{robinson1971comparison}.
An NNI move is defined as a transformation that collapses an interior branch to zero and then expands the resulting degree 4 vertex into a branch in a different way.
The BHV space models the set of trees $\mathcal{T}$ on $n$ species as a cubical complex consisting of a collection of orthants, each isomorphic to $\mathbb{R}_{\ge 0}^{2n-3}$. 
Each orthant of $\mathcal{T}$ corresponds uniquely to a tree topology, and the coordinates in each orthant parameterize the branch lengths for the corresponding tree.
The adjacent orthants of the complex with the same dimension correspond to NNI-adjacent trees.

The BHV space is equipped with a natural metric distance: the shortest path lying in the BHV space between the points. If two points lie in the same orthant, this distance is the usual Euclidean distance.
If two points are in different orthants, they can be joined by a sequence of straight segments, with each segment lying in a single orthant.
We can then measure the length of the path by adding up the lengths of the segments.
The distance between the two trees $\tree_1$ and $\tree_2$ on the BHV space is defined as the minimum of the lengths of such segmented paths joining the two points.

Throughout the paper, we assume that the evolution of the species has not involved reticulate processes, and there exists an underlying ``true species tree" in $\mc{T}$, denoted by $\tree^*$.
Furthermore, $\tree^*$ is a resolved tree with leaf set $\leaves^*$.
In a supertree reconstruction problem, we observe a sequence of gene trees $(\tree_1, \tree_2, \ldots, \tree_k)$, where $\tree_i$ has leaf set $\leaves_i$, and wish to combine these trees into a phylogenetic tree $\hat \tree_k$ on the union of the leaf sets 
\[
\hat \leaves_k = \cup_{i=1}^k \leaves_i \subset \leaves^*.
\] 

In this paper, we consider a tree-generating probability model $\left( \mathcal{P}_{\tree} \right )_{\tree \in \mc{T}}$ such that for each set of species $\leaves \subset \leaves^*$, $\mathcal{P}^{\leaves}_{\tree}(\cdot)$ is a distribution of trees forming by species in $\leaves$.
We assume that the observed gene trees $\{\tree_i\}_{i=1}^k$ are independently distributed according to $\left \{ \mathcal{P}^{\leaves_i}_{\tree^*}(\cdot) \right \}_{i=1}^k$ respectively.
In other words, the joint probability of the observed gene trees is
\[
\mathcal{P}_{\tree^*}(\tree_1, \tree_2, \ldots, \tree_k) = \prod_{i=1}^k{ \mathcal{P}^{\leaves_i}_{\tree^*}(\tree_i) }.
\]
\begin{Remark}
We note that the tree-generating probability model $\mathcal{P}_{\tree^*}$ may not necessarily be nested. 
That is, given two nested sets of leaves $\leaves_1 \subset \leaves_2 \subset \leaves^*$, although the probability $\mathcal{P}^{\leaves_2}_{\tree^*}$ (which is used to generate trees with leaf set $\leaves_2$) also induces a natural distribution on the set of trees with leaf set $\leaves_1$, this probability $\mathcal{P}^{\leaves_2}_{\tree^*} |_{\leaves_1}$ may not be the same as $\mathcal{P}^{\leaves_1}_{\tree^*}$.
\end{Remark}

Given the observed gene trees $(\tree_1, \tree_2, \ldots, \tree_k)$, the ML supertree is defined as
\[
\hat \tree_k = \argmax_{\tree \in \mc{T}}  \ell_k(\tree) 
\]
where $\ell_k(\tree) $ denotes the log-likelihood function
\[
\ell_k(\tree) = \sum_{i=1}^k \log \mathcal{P}^{\leaves_i}_{\tree}(\tree_i).
\]

To enable theoretical analyses of the ML supertree method, we make the following assumptions.

\begin{Assumption}[Weak covering property]
 The sequence of subsets of $\leaves^*$ satisfies the \emph{weak covering property}: there exists $c_1>0, \gamma>1/2$, and $K > 0$ such that for each subset $\leaves$ of taxa from $\leaves^*$ of size $4$,
\[
\frac{1}{k^\gamma}|\{i \le k: \leaves \subset \leaves_i\}| \ge c_1 \h \forall k \geq K.
\]
Here, $|\mc{A}|$ denotes the number of elements in the set $\mc{A}$.
\label{as:cover}
\end{Assumption}

We note that the weak covering property ensures that as the number of gene trees increases, all quartets (subtrees with $4$ leaves) of $\tree^*$ are visited with enough frequency to enable a reliable estimate for each of the quartets. 
Assumption $\ref{as:cover}$ is a direct generalization of the \emph{covering property} introduced in \cite{steel2008maximum}, which can be obtained from Assumption $\ref{as:cover}$ by setting $\gamma=1$. 

\begin{Assumption}[Model identifiability]
For all $\leaves \subset \leaves^*$, the distribution $\mathcal{P}^{\leaves}_{\tree^*}(\cdot)$ is identifiable.
That is, if $d(\tree^* |_{\leaves}, \tree |_{\leaves}) >  0$, then $\text{KL}(\mathcal{P}^{\leaves}_{\tree^*}, \mathcal{P}^{\leaves}_{\tree}) > 0$. 
\label{as:iden}
\end{Assumption}

Assumption $\ref{as:iden}$ guarantees that it is at least possible to reconstruct the restriction $\tree^*|_{\leaves}$ of $\tree^*$ to the subset $\leaves$ with a complete knowledge of $\mathcal{P}^{\leaves}_{\tree^*}$. 
This assumption is similar to, but distinct from, the condition of \emph{basic centrality} introduced by \cite{steel2008maximum}, which requires that for all subsets of  $\leaves \subset \leaves^*$  of size 4, 
\[
\mathcal{P}_{\tree^*}^{\leaves}[\tree^* |_\leaves] \ge (1+\eta) \mathcal{P}_{\tree^*}^{\leaves}[\tree'] 
\]
for all trees $\tree'$ on leaf set $\leaves$ that are different from $\tree^* |_\leaves$, and where $\eta>0$:
\begin{itemize}
\item On one hand, the basic centrality condition implies that if $d(\tree^* |_{\leaves}, \tree |_{\leaves}) >  0$, the modes of the distributions $\mathcal{P}^{\leaves}_{\tree^*}$ and $\mathcal{P}^{\leaves}_{\tree}$ are different. 
From this, we can deduce that $\text{KL}(\mathcal{P}^{\leaves}_{\tree^*}, \mathcal{P}^{\leaves}_{\tree}) > 0$. 
Thus, our identifiability assumption is somewhat weaker than this condition.
Assumption $\ref{as:iden}$ also does not impose any assumption on the family of distribution themselves and thus can be applied to a wider class of probabilistic models. 

\item On the other hand, the basic centrality condition only concerns leaf sets of size 4, while Assumption $\ref{as:iden}$ impose restrictions on leaf sets of all sizes. 
However, we note that conditions on subset of leaves (such as the basic centrality condition) only work under the implicit assumption that there are some connection between $\mathcal{P}^{\leaves_2}_{\tree}|_{\leaves_1}$ and  $\mathcal{P}^{\leaves_1}_{\tree}$ for $\leaves_1 \subset \leaves_2$.
Since our framework does not assume any nested structure in the probability model, Assumption $\ref{as:iden}$ is more appropriate. 
\end{itemize}

Finally, we impose the following regularity conditions on the tree-generating probability model:

\begin{Assumption}
(Regularity)

\begin{itemize}

\item[(a)] For all $\tree \in \mc{T}$, $\leaves \subset \leaves^*$, and any tree $\tree'$ with leaf set $\leaves$, 
\begin{itemize}
\item[i. ] $\mathcal{P}^{\leaves}_{\tree}(\tree') > 0$,
\item[ii. ] $\log \mathcal{P}^{\leaves}_{\tree}(\tree')$ is a locally-Lipchitz function with respect to $\tree$, and the Lipchitz constant does not depend on $\tree'$.
\end{itemize}

\item[(b)] There exist $c_2, c_3 > 0$, $m \ge 2$ such that for any leaf set $\leaves$ and tree $\tree$, if $d(\tree^*, \tree) \le c_2$, then 
\[
\text{KL}(\mathcal{P}^{\leaves}_{\tree^*}, \mathcal{P}^{\leaves}_{\tree}) \ge c_3 d(\tree^* |_{\leaves}, \tree |_{\leaves})^m.
\] 
\end{itemize}
\label{as:regular}
\end{Assumption}

\begin{Remark}
If for each tree $\tree'$ with leaf set $\leaves$, the probability density function $\mathcal{P}^{\leaves}_{\tree}(\tree')$ is an analytic function with respect to $\tree$ in a neighborhood of $\tree^*$, then Assumption $\ref{as:regular}$(b) holds.
\label{rem:smooth}
\end{Remark}

\begin{proof}
Note that the function $\mathbb{E}_{\mathcal{P}_{\tree^*}}\left [\log \mathcal{P}^{\leaves}_{\tree}(\tree')\right]$ (with respect to $\tree$) is analytic in a neighborhood $\mc{U}$ of $\tree^*$. 
Define 
\[
\mc{A} = \{ \tree \in \mc{U}:  \mathbb{E}_{\tree' \sim \mathcal{P}_{\tree^*}}\left [\log \mathcal{P}^{\leaves}_{\tree^*}(\tree')\right] = \mathbb{E}_{\tree' \sim \mathcal{P}_{\tree^*}}\left [\log \mathcal{P}^{\leaves}_{\tree}(\tree')\right] \}.
\]
For all $\tree \in \mc{A}$, we have 
\[
\text{KL}(\mathcal{P}^{\leaves}_{\tree^*}, \mathcal{P}^{\leaves}_{\tree}) = \mathbb{E}_{\tree' \sim \mathcal{P}_{\tree^*}}\left [\log \mathcal{P}^{\leaves}_{\tree^*}(\tree')\right] - \mathbb{E}_{\tree' \sim \mathcal{P}_{\tree^*}}\left [\log \mathcal{P}^{\leaves}_{\tree}(\tree')\right] = 0.
\]
By Assumption \ref{as:iden}, we conclude that $\tree |_\leaves = \tree^* |_\leaves$ for all $\tree \in \mc{A}$.
Applying {\L}ojasiewicz inequality \citep[][Theorem 1]{ji1992global}, we deduce that there exists $c_3 >0$ and $m \ge 2$ such that for any $\tree$ in the neighborhood,
\[
\mathbb{E}_{\mathcal{P}_{\tree^*}}\left [\log \mathcal{P}^{\leaves}_{\tree^*}(\tree')\right] - \mathbb{E}_{\mathcal{P}_{\tree^*}}\left [\log \mathcal{P}^{\leaves}_{\tree}(\tree')\right] \ge  c_3 d(\tree, \mc{A})^m \geq c_3 d(\tree^* |_{\leaves}, \tree |_{\leaves})^m.
\]
This implies the result.
\end{proof}

Throughout this paper, we will assume that Assumptions $\ref{as:cover}$, $\ref{as:iden}$, and $\ref{as:regular}$ hold.
In the next section, we will establish the following convergence rate of the ML supertree.

\begin{thm}
Under Assumptions $\ref{as:cover}$, $\ref{as:iden}$, and $\ref{as:regular}$, for any $\delta>0$, there exist constants $m>0$, $ C_{\delta} > 0$ and $K_\delta>0$ such that for all $k \ge K_\delta$, we have
\[
d(\hat \tree_k , \tree^*)\leq C_{\delta} \left( \frac{\log k}{k^{(\gamma-1/2)}}\right)^{1/m}
\]
with probability at least $1-\delta$.
\label{thm:consistency}
\end{thm}

Theorem \ref{thm:consistency} establishes that under fairly mild regularity conditions, the ML supertree is consistent and has a polynomial convergence rate. 
Notably, the result holds for all identifiable family of locally-analytic distributions (Remark \ref{rem:smooth}). 
The degree of this (polynomial) convergence rate depends on the sampling scheme of the leaves (characterized by the covering coefficient $\gamma$) and on the geometry of the probabilistic model (characterized by the constant $m$ in Assumption $\ref{as:regular}$(b)). 
We further note that Assumption $\ref{as:regular}$(b)) is only required for establishing the convergence rate of the ML estimator and the absence of this condition does not affect the proof of consistency.

\begin{cor}
Under Assumptions $\ref{as:cover}$, $\ref{as:iden}$, and $\ref{as:regular}(a)$, the ML supertree is consistent.
\end{cor}

\begin{Remark}
Theorem $\ref{thm:consistency}$ is still valid if the universal Assumption $\ref{as:regular}$(b) is replaced by the following condition:

\emph{
There exist $c_2, c_3 > 0$, $m \ge 2$ such that if $d(\tree^*, \tree) \le c_2$, then 
\[
\text{KL}(\mathcal{P}^{\leaves_i}_{\tree^*}, \mathcal{P}^{\leaves_i}_{\tree}) \ge c_3 d(\tree^* |_{\leaves_i}, \tree |_{\leaves_i})^m \h \forall i \in \mathbb{N}.
\] 
}
\label{rem:reg}
\end{Remark}

\section{Convergence of maximum likelihood supertree}

To enable the analysis of convergence, we define
\[
R_{k}(\tree) = \frac{1}{k}\ell_k(\tree^*) - \frac{1}{k}\ell_k(\tree) 
\]
and 
\[
\mc{R}_k(\tree) := E_{\mathcal{P}_{\tree^*}}[R_{k}(\tree)] =  \frac{1}{k} \left(E_{\mathcal{P}_{\tree^*}}[\ell_k(\tree^*)]  - E_{\mathcal{P}_{\tree^*}}[\ell_k(\tree)]  \right) = \frac{1}{k} \text{KL}(\mathcal{P}_{\tree^*}, \mathcal{P}_{\tree}).
\]
We refer to $R_{k}$ and $E_{\mathcal{P}_{\tree^*}}[R_{k}(\tree)]$ as the \emph{empirical risk function} and the \emph{expected risk function}, respectively. 

An intuitive argument for the consistency of the ML estimator can be described as follows. 
As the number of gene trees increases, we have
\[
\left |R_{k}(\tree)  - \mc{R}_k(\tree) \right|  \to 0 
\]
with sufficiently high probability. 
Therefore, the ML estimator will converge to the optimal value of the risk function, which is attained at the true species tree $\tree^*$.
This simple argument is formalized by a lower bound of the expected risk function (Section 3.1) and a uniform concentration bound on the deviation of the empirical risk function and its expectation (Section 3.2). 

\subsection{Lower bound of the expected risk}

\begin{Lemma}

There exist a neighborhood $\mc{U}$ of $\tree^*$ and $c_4>0$ such that 
\[
\mathbb{E}_{\mathcal{P}_{\tree^*}}\left [\ell_{k}(\tree^*) - \ell_{k}(\tree)  \right] \ge c_4 k^{\gamma} d(\tree^*, \tree)^m \h \forall \tree \in \mc{U}.
\]
\label{info}
\end{Lemma}

\begin{proof}

We note that
\[
\mathbb{E}_{\mathcal{P}_{\tree^*}}\left [\ell_{k}(\tree^*) - \ell_{k}(\tree)  \right]  = \text{KL}(\mathcal{P}_{\tree^*}, \mathcal{P}_{\tree}).
\]
Consider an arbitrary leaf set $\leaves \subset \leaves^*$ of size 4 and define $\mc{I}= \{i \le k: \leaves \subset \leaves_i\}$.
By the weak covering property, we have $|\mc{I}| \ge c_1 k^{\gamma}$.
Thus,
\[
 \text{KL}(\mathcal{P}_{\tree^*}, \mathcal{P}_{\tree})  \ge \sum_{i \in \mc{I}}{\text{KL}(\mathcal{P}^{\leaves_i}_{\tree^*}, \mathcal{P}^{\leaves_i}_{\tree})} \geq \sum_{i \in \mc{I}}{c_3 d(\tree^*|_{\leaves_i}, \tree|_{\leaves_i})^m} \ge c_1 c_3 k^{\gamma} (2n-3)^{m/2} d(\tree^*|_\leaves, \tree|_\leaves)^m
\]
for all $\tree$ in some neighborhood $U$ around $\tree^*$. 
Here, the second inequality comes from Assumption $\ref{as:regular}$(b).

On the other hand, \cite{dinh2018consistency} (Lemma 6.2 (i)) proved that for some leaf set $\leaves \subset \leaves^*$ of size 4, we have
\[
d(\tree^*|_\leaves, \tree|_\leaves) \ge \frac{1}{(2n-3)} d(\tree^*, \tree).
\]
We deduce that
\[
\text{KL}(\mathcal{P}_{\tree^*}, \mathcal{P}_{\tree}) \ge \frac{1}{(2n-3)^{m/2}} c_1 c_3 k^{\gamma} d(\tree^*, \tree)^m
\]
\end{proof}

\begin{Lemma}
For any leaf set $\leaves$ and any $x>0$, there exists $C_x>0$ such that
\[
 \text{KL}(\mathcal{P}^{\leaves}_{\tree^*}, \mathcal{P}^{\leaves}_{\tree}) \ge C_x 
\]
when $d(\tree^*|_\leaves, \tree|_\leaves) \ge x$.
\label{prelimc}
\end{Lemma}

\begin{proof}
Assume that there exists a sequence of tree $\{\tree_i\}$ such that $d(\tree^*|_\leaves, \tree_i|_\leaves) \ge x$ and $ \text{KL}(\mathcal{P}^{\leaves}_{\tree^*}, \mathcal{P}^{\leaves}_{\tree_i}) \to 0$. 
Since the tree space is compact, we can extract a sub-sequence $\{\tree_{i_j}\}$ of $\{\tree_i\}$ that converges to some tree $\tree_0$. 
We deduce that $ \text{KL}(\mathcal{P}^{\leaves}_{\tree^*}, \mathcal{P}^{\leaves}_{\tree_0}) =0$ and $\tree^*|_\leaves \ne \tree_0|_\leaves$. 
This contradicts Assumption $\ref{as:iden}$.
\end{proof}

\begin{Lemma}
Let $\mc{V}$ be a neighborhood of $\tree^*$. 
There exist $c_{\mc{V}} > 0$ such that for any $\tree \not \in \mc{V}$, we have 
\[
\text{KL}(\mathcal{P}_{\tree^*}, \mathcal{P}_{\tree}) \ge c_{\mc{V}} k^{\gamma}.
\] 
\label{prelimb}
\end{Lemma}

\begin{proof}
By \cite{dinh2018consistency} (Lemma 6.2 (i)), there exists $\leaves \subset \leaves^*$ of size 4 such that
\[
d(\tree^*|_\leaves, \tree|_\leaves) \ge \frac{1}{(2n-3)} d(\tree^*, \tree) \ge C_{\mc{V}} >0
\]
since $\tree \not \in \mc{V}$. 
Define $\mc{I}= \{i \le k: \leaves \subset \leaves_i\}$, we note that for all $i \in \mc{I}$
\[
d(\tree^*|_{\leaves_i}, \tree|_{\leaves_i}) \ge (2n-3)^{1/2} d(\tree^*|_\leaves, \tree|_\leaves) \ge C'_{\mc{V}},
\]
which implies $ \text{KL}(\mathcal{P}^{\leaves_i}_{\tree^*}, \mathcal{P}^{\leaves_i}_{\tree}) \ge C''_{\mc{V}}$ for some $ C''_{\mc{V}} >0$ by Lemma $\ref{prelimc}$. 

By the weak covering property, we have $|\mc{I}| \ge c_1 k^{\gamma}$.
Thus,
\[
 \text{KL}(\mathcal{P}_{\tree^*}, \mathcal{P}_{\tree})  \ge \sum_{i \in \mc{I}}{\text{KL}(\mathcal{P}^{\leaves_i}_{\tree^*}, \mathcal{P}^{\leaves_i}_{\tree})} \ge c_1 k^{\gamma} C''_{\mc{V}}
\]
which completes the proof.
\end{proof}

\subsection{Uniform concentration bound}

\begin{Lemma}[Concentration bound]
For any $\delta>0$, $k \ge 3$, there exists $c_5(\delta)$ such that
\[
\left |R_{k}(\tree)  - \mathbb{E}_{\mathcal{P}_{\tree^*}}\left [R_{k}(\tree) \right] \right| \le \frac{c_5 \log k} {\sqrt{k}} \h \forall \tree, 
\]
with probability at least $1-\delta$.
\label{lem:variance}
\end{Lemma}

\begin{proof}
Since the functions $\log \mathcal{P}^{\leaves_i}_{\tree}(\tree_i)$ is locally Lipschitz with respect to $\tree$ (Assumption $\ref{as:regular}$(a)) and number of species is finite, there exists $C_1 > 0$ such that
\[
|\log \mathcal{P}^{\leaves_i}_{\tree}(\tree_i) -  \log \mathcal{P}^{\leaves_i}_{\tree'}(\tree_i)| \le C_1 d(\tree, \tree') \h \forall i, \tree_i, \tree, \tree'.
\]
Since the tree space is compact, $C_1 d(\tree, \tree')$ is bounded by a constant $C_2$.
Using Hoeffding's inequality \citep{hoeffding1963probability}, we obtain
\[
\mathbb{P}\left[ \left|R_{k}(\tree)  - \mathbb{E}_{\mathcal{P}_{\tree^*}}\left [R_{k}(\tree) \right]   \right|  \ge \frac{x}{\sqrt{k}} \right] \le 2\exp\left(-\frac{2x^2}{kC_2^2}\right).
\]
On the other hand, we have $|R_{k}(\tree) - R_k(\tree')| \le C_1 d(\tree, \tree')$ and $|\mathbb{E}_{\mathcal{P}_{\tree^*}}\left [R_{k}(\tree) \right] - \mathbb{E}_{\mathcal{P}_{\tree^*}}\left [R_{k}(\tree') \right]| \leq C_1 d(\tree, \tree')$.
Thus, if we define the events
\[
\mc{A}(x, k, \tree)=\left\{ \left|R_{k}(\tree)  - \mathbb{E}_{\mathcal{P}_{\tree^*}}\left [R_{k}(\tree) \right]   \right|  \ge \frac{x}{2\sqrt{k}} \right\}
\]
and
\[
\mc{B}(x, k, \tree)=\left\{ \exists \tree': d(\tree, \tree') \le \frac{x}{4C_1\sqrt{k}} ~ \text{and} ~ \left|R_{k}(\tree')  - \mathbb{E}_{\mathcal{P}_{\tree^*}}\left [R_{k}(\tree') \right]   \right|  \ge \frac{x}{\sqrt{k}} \right\}
\]
then $\mc{B}(x, k, \tree) \subset \mc{A}(x, k, \tree)$.
Note that the total number of balls of radius $x/(4C_1\sqrt{k})$ required to cover the tree space is bounded above by
\[
C_n \left(\frac{4g_0C_1 \sqrt{k}}{x}\right)^{2n-3} (2n-3)!!
\]
where $g_0$ is the upper bound of branch lengths and $C_n$ is a constant depending on the number of species.
To obtain the desired inequality, we will chose $x$ such that
\[
C_n \left(\frac{4g_0C_1 \sqrt{k}}{x}\right)^{2n-3} (2n-3)!! \times 2\exp\left(-\frac{2x^2}{kC_2^2}\right) \le \delta 
\] 
which can be done with $x = C(\delta, n, C_2, g_0) \log k$.
\end{proof}

\subsection{Proof of Theorem \ref{thm:consistency}}


First, we establish that the ML estimator is consistent. 

By Lemma $\ref{lem:variance}$, we have
\[
\mc{R}(\hat \tree_k) \le  R_k(\hat \tree_k) + c_5 \frac{ \log(k)}{\sqrt{k}} \le  c_5 \frac{ \log(k)}{\sqrt{k}} 
\]
with probability at least $1-\delta$, since $\hat \tree_k$ is the maximizer of the empirical risk function. 

On the other hand, let $\mc{V}$ be a neighborhood of $\tree^*$, by Lemma \ref{prelimb},  we have
\[
\mathbb{E}_{\mathcal{P}_{\tree^*}} [R_k(\tree')] \ge c_{\mc{V}} \frac{1}{k^{1-\gamma}}\h \forall \tree' \not \in \mc{V}.
\]
Since $\gamma>1/2$, there exists $K_{\delta, \mc{V}}$ such that for $k \ge K_{\delta, \mc{V}}$, we have 
\[
c_5 \frac{ \log(k)}{\sqrt{k}}   \le c_{\mc{V}} \frac{1}{k^{1-\gamma}}
\]
We deduce that for $k \ge K_{\delta, \mc{V}}$, $\hat \tree_k \in \mc{V}$ with probability at least $1 -\delta$. 
This proves that the ML estimator is consistent. 

Next, we will derive the convergence rate of the ML supertree. 
From Lemma \ref{info}, there exist a neighborhood $\mc{U}$ of $\tree^*$ and $c_4, m>0$ such that 
\[
\mathbb{E}_{\mathcal{P}_{\tree^*}}[R_k(\tree)] = \frac{1}{k}\mathbb{E}_{\mathcal{P}_{\tree^*}}\left [\ell_{k}(\tree^*) - \ell_{k}(\tree)  \right] \ge c_4 \frac{1}{k^{1-\gamma}}\ d(\tree^*, \tree)^m \h \forall \tree \in \mc{U}.
\]
For $k \ge K_{\delta, \mc{U}}$, we have $\hat \tree_k \in \mc{U}$ with probability at least $1 -\delta$.
By Lemma \ref{lem:variance}, with probability $1 - 2\delta$, we obtain
\begin{align*}
c_4 \frac{1}{k^{1-\gamma}}\ d(\tree^*, \hat \tree_k)^m & \le \mc{R}_k(\hat \tree_k) -  R_k(\hat \tree_k)   \\
& \le c_5 \frac{ \log(k)}{\sqrt{k}}.
\end{align*}
Here, 
\[
R_k(\hat \tree_k) = \frac{1}{k}\ell_k(\tree^*) - \frac{1}{k}\ell(\hat \tree _k) \le 0
\]
because $\hat \tree_k$ is the maximizer of $R_k$. 
This completes the proof.

\section{Applications: convergence under the exponential model}

The \emph{exponential model} is a simple model of gene tree estimation errors in which the probability of observing a given tree decreases exponentially with its distance from the species tree \citep{steel2008maximum}. 
Suppose $d$ is some metric on the set of trees, in the exponential model, the probability of reconstructing any tree $\tree'$ with a leaf set $\leaves$, when $\tree^*$ is the generating tree is proportional to an exponentially decaying function of the distance from $\tree'$ to $\tree^*|_{\leaves}$:
\[
\mathcal{P}^{\leaves}_{\tree^*}(\tree') = \alpha_{\tree^*, \leaves} \exp (-\beta_{\leaves} d(\tree', \tree^* |_{\leaves})).
\] 
Here $\beta_{\leaves}$ is a constant that depends only on the set of leaves $\leaves$, while $\alpha_{\tree^*, \leaves}$ is the normalizing constant to ensure that $\mathcal{P}^{\leaves}_{\tree^*}$ is a density function.

In this section, we verify the identifiability and regularity conditions for the exponential model when $d$ is any continuous tree distance such that if $\tree$ and $\tree'$ have the same topology, then $d(\tree, \tree')$ is the Euclidean distance.

\begin{thm}
Under the exponential model with Assumptions $\ref{as:cover}$, for any $\delta>0$, $ C_{\delta} > 0$ and $K_\delta > 0$ such that for all $k \ge K_\delta$,
\[
d(\hat \tree_k , \tree^*)\leq C_{\delta} \left( \frac{\log k}{k^{(\gamma-1/2)}}\right)^{1/m}
\]
with probability at least $1-\delta$, where 
\[
m = 4\sup_{i}{|\leaves_i|} -4. 
\]
\label{thm:exp}
\end{thm}

\begin{proof}

First, we note that for all $\tree'$ with leaf set $\leaves$, $\mathcal{P}^{\leaves}_{\tree}(\tree')$ is a positive and locally Lipschitz function. 
To obtain an upper bound on the convergence rate for the ML supertree, we need to verify Assumption $\ref{as:iden}$ and Assumption $\ref{as:regular}(b)$/Remark $\ref{rem:reg}$.

\smallskip

\textbf{Identifiability}. Since it is sufficient to prove either $\text{KL}(\mathcal{P}^{\leaves}_{\tree^*}, \mathcal{P}^{\leaves}_{\tree}) > 0$ or $\text{KL}(\mathcal{P}^{\leaves}_{\tree}, \mathcal{P}^{\leaves}_{\tree^*}) > 0$, we can assume that $\alpha_{\tree^*, \leaves} \geq \alpha_{\tree, \leaves}$. 

We denote 
\[
\mc{A}_c = \{ \tree' : \beta_{\leaves} d(\tree', \tree^* |_{\leaves}) \leq  c \}
\] 
and pick $c$ small enough such that $\beta_{\leaves} d(\tree', \tree |_{\leaves}) \ge 2c$.
We note that in $A_c$,
\[
 \exp \left(-\beta_{\leaves} d(\tree', \tree^* |_{\leaves}) \right)  \ge \exp \left(-\beta_{\leaves} d(\tree', \tree |_{\leaves}) \right).
\]
Since $\alpha_{\tree^*, \leaves} \geq \alpha_{\tree, \leaves}$, we have
\begin{align*}
 \alpha_{\tree^*, \leaves} \exp \left(-\beta_{\leaves} d(\tree', \tree^* | \leaves) \right) &-  \alpha_{\tree, \leaves} \exp \left(-\beta_{\leaves} d(\tree', \tree | \leaves) \right)\\
 &\ge  \alpha_{\tree, \leaves} \left[ \exp \left(-\beta_{\leaves} d(\tree', \tree^* | \leaves) \right)  -\exp \left(-\beta_{\leaves} d(\tree', \tree^* | \leaves) \right)  \right].
\end{align*}
By Pinsker's inequality,
\begin{align*}
\text{KL}(\mathcal{P}^{\leaves}_{\tree^*}, \mathcal{P}^{\leaves}_{\tree}) &\ge 2d_{TV}(\mathcal{P}^{\leaves}_{\tree^*}, \mathcal{P}^{\leaves}_{\tree}) ^2 \\
& \geq 2  \alpha_{\tree, \leaves}^2  \left [ \int_{\mc{A}_c}{ \exp \left(-\beta_{\leaves} d(\tree', \tree^* | \leaves) \right)- \exp \left(-\beta_{\leaves} d(\tree', \tree | \leaves) \right)}  \right ]^2 \\
& \geq  2 \alpha_{\tree, \leaves}^2 \mu(\mc{A}_c)^2 (e^c - e^{2c})^2>0,
\end{align*}
which establishes the identifiability of the exponential model.
Here, $\mu$ is the Lebesgue measure.

\bigskip

\textbf{Regularity}. 
Consider a fixed tree $\tree \in \mathcal{T}$ and leaf set $\leaves$, we first assume that $\alpha_{\tree^*, \leaves} \geq \alpha_{\tree, \leaves}$.
If we define
\[
\mc{A} = \{ \tree' : \beta_{\leaves} d(\tree', \tree^* |_{\leaves}) \leq  \frac{1}{3}d(\tree |_{\leaves}, \tree^* |_{\leaves})) \},
\] 
then
\[
\beta_{\leaves} d(\tree', \tree |_{\leaves}) \ge  \frac{2}{3}d(\tree |_\leaves, \tree^* |_{\leaves}) ~~~ \forall \tree' \in \mc{A}.
\] 
Using the same argument as above, we have
\begin{align*}
d_{TV}(\mathcal{P}^{\leaves}_{\tree^*}, \mathcal{P}^{\leaves}_{\tree}) & \geq  \alpha_{\tree, \leaves_i} \mu(A) \left [e^{-\frac{1}{3} d(\tree |_\leaves, \tree^* |_{\leaves})} - e^{-\frac{2}{3}d(\tree |_\leaves, \tree^* |_{\leaves})} \right].
\end{align*}
It can be verified that there exists $C > 0$ such that for all $x \in [0, C]$, we have
\begin{equation}
e^{-x} - e^{-2x} \ge \frac{x}{2}.
\label{eq:ex}
\end{equation}
Moreover since $d$ is Euclidean inside each orthant, if $d(\tree, \tree^*)$ is small enough, we have
\begin{equation}
\mu(A) = C_{2|\leaves| - 3} \left(\frac{d(\tree |_\leaves, \tree^* |_{\leaves})}{ 3 \beta_{\leaves}}\right)^{2 |\leaves| -3}.
\label{eq:mu}
\end{equation}
Thus, let $\mc{W}$ be a neighborhood in the same topology of $\tree^*$ such that
\[
\alpha_{\tree, \leaves}  \ge \frac{1}{2} \alpha_{\tree^*, \leaves} \h \forall \tree \in \mc{W},
\]
and both Equations $\eqref{eq:ex}$ and $\eqref{eq:mu}$ are satisfied, we have
\[
\text{KL}(\mathcal{P}^{\leaves}_{\tree^*}, \mathcal{P}^{\leaves}_{\tree}) \ge C d(\tree, \tree^* |_{\leaves})^{4|\leaves| - 4} \h  \forall \tree \in \mc{W},
\]
for some constant $C>0$ independent of $\tree$. 

For the case when $\alpha_{\tree^*, \leaves} < \alpha_{\tree, \leaves}$ we define
\[
\mc{A}'= \{ \tree' : \beta_{\leaves} d(\tree', \tree^* |_{\leaves}) \leq  \frac{1}{3}d(\tree'|_{\leaves}, \tree|_{\leaves})) \},
\] 
and the argument proceeds similarly. 
This validates Assumption $\ref{as:regular}(b)$ for exponential model with $m=4n-4$. 
However, if we use the regularity condition in Remark $\ref{rem:reg}$, we can obtain the result with
\[
m = 4\sup_{i}{|\leaves_i|} -4. 
\]
\end{proof}

\section{Discussion and Conclusion}

In this paper, we propose a novel analytic approach to analyze the convergence of the ML supertree method. 
Instead of focusing on reconstructing the correct discrete topology of the species tree as in previous studies \citep[e.g.][]{roch2015robustness,steel2008maximum}, we employ a continuous model of the tree space and analyze the ML estimator on this metric space, aiming at recovering both the topology and the branch lengths of the species tree. 
This framework enables us to use tools from statistical learning and information theory to establish the convergence rate of the ML estimator and at the same time, to weaken the conditions to obtain consistency and convergence of the estimator.
Our \emph{weak covering property} is an extension of the classical covering property \citep{steel2008maximum} and provides a considerable relaxation on the sampling schemes for supertree estimation.
Our identifiability condition is also more intuitive and generalizable than the well-known basic centrality condition and does not impose constraints on the shape of the probabilistic model of gene tree estimation errors. 
Our information-theoretical approach to analyze statistical estimator on tree spaces is of independent interest and can be extended to other problems in phylogenetics.

There are several avenues for future directions for this work. 
The first direction is extending our results to other practical models of phylogenetic errors, including the multiple-coalescent model (along with a detailed model of the effects of short sequence length on the accuracy in estimating the individual gene trees).
Second, while our result provides a polynomial bound on the convergence rate, the power of the convergence (characterized by the geometric constant $m$) is not sharp. 
A sharper bound of the convergence rate would be of great interest to the field (from both theoretical and applied perspective) and would require further understanding of the tree-generating probabilistic model. 

\section*{Acknowledgement}
LSTH was supported by startup funds from Dalhousie University, the Canada Research Chairs program, the
NSERC Discovery Grant RGPIN-2018-05447, and the NSERC Discovery Launch Supplement DGECR-2018-00181.
VD was supported by a startup fund from University of Delaware and National Science Foundation grant DMS-1951474.

\bibliographystyle{chicago}
\bibliography{biblio}

\end{document}